\documentclass[conference]{IEEEtran}
\IEEEoverridecommandlockouts
\usepackage{cite}
\usepackage{amsmath,amssymb,amsfonts,amsthm}
\usepackage{algorithmic}
\usepackage{graphicx}
\usepackage{mathtools}
\usepackage{numprint}
\usepackage{textcomp}
\usepackage{url}
\usepackage{xcolor}
\usepackage{xfrac}
\def\BibTeX{{\rm B\kern-.05em{\sc i\kern-.025em b}\kern-.08em
    T\kern-.1667em\lower.7ex\hbox{E}\kern-.125emX}}

\npdecimalsign{\ensuremath{.}}

\DeclarePairedDelimiter\floor{\lfloor}{\rfloor}

\newtheorem{theorem}{Theorem}

\newtheorem{proposition}{Proposition}

\newcommand{\mybinom}[3][0.8]{\scalebox{#1}{$\dbinom{#2}{#3}$}}

\makeatletter
{\small 
\xdef\f@size@small{\f@size}
\xdef\f@baselineskip@small{\f@baselineskip}
\normalsize 
\xdef\f@size@normalsize{\f@size}
\xdef\f@baselineskip@normalsize{\f@baselineskip}
}
\newcommand{\smalltonormalsize}{%
  \fontsize
    {\fpeval{(\f@size@small+\f@size@normalsize)/2}}
    {\fpeval{(\f@baselineskip@small+\f@baselineskip@normalsize)/2}}%
  \selectfont
}
\makeatother

\makeatletter
{\footnotesize 
\xdef\f@size@footnotesize{\f@size}
\xdef\f@baselineskip@footnotesize{\f@baselineskip}
\small 
\xdef\f@size@small{\f@size}
\xdef\f@baselineskip@small{\f@baselineskip}
}
\newcommand{\footnotesizetosmall}{%
  \fontsize
    {\fpeval{(\f@size@footnotesize+\f@size@small)/2}}
    {\fpeval{(\f@baselineskip@footnotesize+\f@baselineskip@small)/2}}%
  \selectfont
}
\makeatother
    
\begin{document}

\title{Meta Distribution of Passive Electromagnetic Field Exposure in Cellular Networks\\

\thanks{This work was supported by Innoviris under the Stochastic Geometry Modeling of Public Exposure to EMF (STOEMP-EMF) grant.}
}

\author{Quentin~Gontier$^{* \dagger}$, Charles~Wiame$^{\ddagger}$,
François~Horlin$^{*}$, \\
Christo~Tsigros$^{\ddagger}$, Claude~Oestges$^{\mathsection}$,
Philippe~De~Doncker$^{*}$\\
$^{*}$ Universit\'e Libre de Bruxelles, Brussels, Belgium, \\
Email: \{quentin.gontier, francois.horlin, philippe.dedoncker\}@ulb.be\\
$^{\dagger}$ Brussels Environnement,  Brussels, Belgium, Email: \{qgontier, ctsigros\}@environnement.brussels\\
$^{\ddagger}$ Massachusetts Institute of Technology, Cambridge, USA, Email: wiame@mit.edu\\
$^{\mathsection}$ Universit\'e Catholique de Louvain, Louvain-la-Neuve, Belgium, Email: claude.oestges@uclouvain.be\\}


\maketitle

\begin{abstract}
This paper focuses on the meta distribution of electromagnetic field exposure (EMFE) experienced by a passive user in a cellular network implementing dynamic beamforming. The meta distribution serves as a valuable tool for extracting fine-grained insights into statistics of individual passive user EMFE across the network. A comprehensive stochastic geometry framework is established for this analysis. Given the pivotal role of accurately modeling the main and side lobes of antennas in this context, a multi-cosine gain model is introduced. The meta distribution is closely approximated by a beta distribution derived from its first- and second-order moments, which is demonstrated to be mathematically tractable. The impact of the number of antennas in the ULA on the meta distribution is explored, shedding light on its sensitivity to this parameter.
\end{abstract}

\begin{IEEEkeywords}
Dynamic beamforming, EMF exposure, meta distribution, Nakagami-$m$ fading, stochastic geometry.
\end{IEEEkeywords}

\section{Introduction}
\subsection{Motivation}
The rapid evolution of wireless communication technologies has prompted increasing concerns regarding potential risks associated with electromagnetic field exposure (EMFE) from wireless network infrastructures. Among the effects associated with non-ionizing frequencies, thermal effects are prominently recognized within the scientific community \cite{rumney19}. Organizations like the International Commission on Non-Ionizing Radiation Protection (ICNIRP) recommends limits on EMFE in terms of specific absorption rate or incident power density (IPD), based on conservative margins and reviews of the literature \cite{icnirp2020}. Subsequently, national or regional entities may adopt these guidelines or enact more stringent legislation.

A significant innovation introduced by 5G and beyond 5G networks is dynamic beamforming (DBF). By building arrays of antennas at the base station (BS), DBF facilitates the formation of narrow beams to mitigate interference. It results in a higher mean EMFE for \textit{active users} calling the beam, compared to \textit{passive users} (PUs) who are inactive in the network and typically experience lower EMFE and shorter exposure times, on average \cite{ANFR2019a,Chiaraviglio21}.

To evaluate the EMFE experienced by a random PU in the network, the commonly analyzed metric is the cumulative distribution function (CDF) of EMFE, denoted as $F_{emf}(T_e) = \mathbb P\left[\mathcal{S} < T_e\right]$, where $\mathcal{S}$ is the IPD and $T_e$ represents the EMFE upper threshold, which is expressed in W/m$^2$ within this paper, though it can be effortlessly converted to V/m \cite{GontierAccess}. This metric can be computed using stochastic geometry (SG), with the locations of BSs modeled as point processes (PPs) \cite{GontierAccess, GontierTWC}. 

Three sources of randomness contribute to the received power from BSs. First, the path loss attenuation depending on the PU's location in the network or, equivalently, on the locations of BSs corresponding to a single realization of the PP. Second, time fading in the propagation channel. Third, changes in the orientation of the beams. The first source is linked to spatial aspects of the SG model, while the latter two are tied to small-scale time-varying fluctuations. Although the CDF of EMFE indicates whether EMFE is below $T_e$, it fails to differentiate between temporal and spatial effects, thus offering limited insight into individual EMFE and neglecting potential increases in EMFE due to DBF for shorter periods. To isolate the temporal aspects from the spatial aspects, the CDF of EMFE conditioned on one PP realization $\Psi$ is introduced, and denoted as $F_{emf}(T_e|\Psi) = \mathbb P\left[\mathcal{S} < T_e|\Psi\right]$. For $s \in [0, 1]$, the meta distribution of EMFE is then defined as
\begin{equation}\label{eq:meta_def}
\mathcal{F}_{F{emf}}(T_{e}, s) = \mathbb P_\Psi \left[F_{emf}(T_e|\Psi) > s\right].
\end{equation}
Here, $\mathbb P_\Psi[\cdot]$ signifies that the probability operator is taken over all spatial realizations of the PP. When considering a large network, selecting a new realization of the PP around the centric PU at (0,0) is akin to observing the realization of the PP from the point of view of a distant PU in the network. Consequently, employing $\mathbb{P}_{\Psi}$ essentially evaluates performance for numerous users uniformly distributed across the network. In terms of interpretation, $\mathcal{F}_{F{emf}}(T_{e}, s) = x$\% means that $x$\% of PUs in the network experience an EMFE below the threshold $T_e$ for at least a fraction $s$ of the time. 

\subsection{Related Works}



\paragraph{DBF Models in SG}
Antenna pattern models aim to approximate the theoretical distribution of the antenna factor of a Uniform Linear Array (ULA), which is mathematically intractable. The widely used flat-top pattern, despite its lack of realism, maintains mathematical tractability in SG by assuming a flat gain for main lobes and lower flat gain for side lobes \cite{direnzo15}. An alternative, the cosine approximation, closely matches theoretical distributions, but does not model side lobes \cite{yu2017}. Finally, the Gaussian approximation accurately models the main lobe and maintains a flat gain for the side lobes \cite{Thornburg15}. 

\paragraph{EMFE Assessment Using SG}

In the recent years, SG has been crucial in evaluating IPD to optimize Wireless Power Transfer (WPT) or minimize it for EMF-aware applications. BSs are commonly modeled as homogeneous Poisson PP (PPPs) in SG models, facilitating computational tractability while maintaining accuracy. The flat-top gain model is widely used in WPT systems for energy harvesting \cite{Khan16} and energy correlation analysis \cite{ECC}. Furthermore, the Gaussian model in \cite{Wang21} examines the effects of imperfect beam alignment in millimeter wave (mmWave) WPT systems. Initial EMFE assessments in the SG framework are detailed in \cite{app10238753}, using an empirical model for 5G mmWave scenarios, and in \cite{GontierAccess}, comparing theoretical and experimental EMFE distributions in urban areas. The flat-top pattern is applied for EMFE analysis in coexisting sub-6~GHz and mmWave networks in \cite{9511258}, as well as a for a joint EMFE and SINR analysis in $\beta$-Ginibre PPs and inhomogeneous PPPs in \cite{GontierTWC}. 

\paragraph{Meta Distributions in SG}
The concept of conditional success probability, introduced in \cite{Ganti10}, isolates the user orientation's influence in a Poisson bipolar network by using the complementary cumulative distribution function of SINR conditioned on a PP realization. In \cite{Haenggi16}, the meta distribution of SINR is formally defined for the same PP, providing both an exact mathematical derivation and an approximation using a beta distribution. Extensive research, such as \cite{uplink_meta, Quan22}, applies the meta-distribution for coverage analysis, computing mean local delay and investigating the impact of DBF with beam misalignment, employing the flat-top model. The SINR meta distribution is also used for spatio-temporal analysis in mmWave heterogeneous networks with varying queue status~\cite{Yang22}.

\subsection{Contributions}

A notable gap exists in the literature concerning the study of EMFE experienced by a random PU in a DBF network, particularly within the SG framework. To the best of the authors' knowledge, no prior work has simultaneously studied the spatial and temporal fluctuations of EMFE. Moreover, all models used in the SG framework only consider one antenna array per BS, covering 360$^\circ{}$ and struggle to accurately capture the characteristics of the side lobes, which affects the CDF of EMFE of the PU. Motivated by these considerations, the contributions of this paper are as follows.
\begin{itemize}
    \item We introduce a novel performance metric for EMFE assessment in cellular networks, the meta distribution of EMFE. This metric is statistically finer than existing ones, as it takes into account both spatial and temporal aspects. In addition, we provide a mathematical expression for closely approximating this metric in the case of a PU in a network with DBF and Nakagami-$m$ fading.
    \item An accurate description of the gain pattern of BSs featuring 3 identical colocated ULAs, each serving a different sector, is needed to obtain a realistic description of the meta distribution of EMFE in a network employing DBF. Therefore, we introduce in this paper a novel antenna gain model named multi-cosine model.
\end{itemize}


\section{System Model}
\label{sec:system_model}

\subsection{Topology}
\label{ssec:topology}
Let the two-dimensional spatial domain $\mathcal{B} \in \mathbb{R}^2$ be the network area, defined as a disk with a radius $\tau$ and centered at the origin. Within $\mathcal{B}$, let $\Psi = {X_i}$ denote the PPP representing the locations of BS $X_i$, all sharing the same technology, belonging to the same network provider, operating at a carrier frequency $f$, and being able to transmit at a maximum power $P_t$. The density of BSs in $\Psi$ is denoted by $\lambda$. Each BS is situated at a height $z > 0$ relative to the users. The PU is positioned at the origin. As shown in Fig.~\ref{fig:network_Globecom}, the boresight direction of $X_i$ forms an angle $\xi_i$ with the PU. Each BS is equipped with three identical ULAs oriented at 120$^\circ$ intervals as per 3GPP specifications, comprising $N$ antennas with half-wavelength spacing. Consequently, the angle $\xi_i$ is a random variable in the interval $[0,2\pi/3[$. Intra-cell interference is neglected for simplicity, and an exclusion radius $r_e$ around the user ensures no BS is located within this region. The normalized gain $G(\xi)$ is uniformly scaled by the maximum gain $G_{\textrm{\normalfont max}} = N$, and defined in Subsection~\ref{ssec:gain}. For a conservative approach, the network is considered fully loaded, with each ULA communicating with one user.

\begin{figure}
    \centering
    \includegraphics[width=0.8\linewidth, trim={1cm, 4.5cm, 1cm, 1cm}, clip]{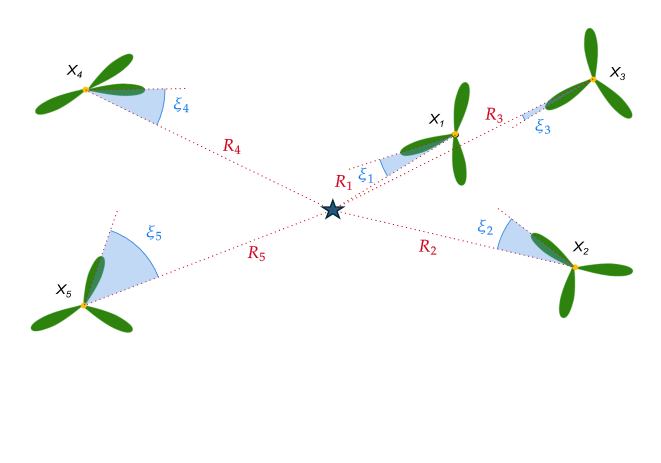}
    \caption{Network topology. The PU's location is represented by the star, the BSs are represented with the three ULA gains.}
    \label{fig:network_Globecom}
\end{figure}

\subsection{Propagation Model}
\label{ssec:propagation}

The propagation model is defined as
{\smalltonormalsize
\begin{equation}\label{eq:model}
    P_{r, i} = P_t G_i |h_i|^2 l_i 
\end{equation}}
where $P_{r, i}$ is the received power from BS $X_i$, $G_i$ is the BS gain towards the PU, $|h_i|^2$ accounts for the fading and $l_i$ is the path loss attenuation. Specifically, $l_i = l(X_i) = \kappa^{-1} \left(R_i^2+z^2\right)^{-\alpha/2}$ for a distance $R_i$ between $X_i$ and the PU, with $\alpha > 2$ the path loss exponent and $\kappa = (4\pi f/c_0)^2$ where $c_0$ is the speed of light. The channel $h_i$ follows a Nakagami-$m$ fading model, making $|h_i|^2$ gamma-distributed with shape parameter $m$ and scale parameter $1/m$. Consequently, the CDF of $|h_i|^2$ is expressed as $F_{|h|^2}(x) = \gamma(m,m x)/\Gamma(m)$. Utilizing the compact notation $\Bar{S}_{r}(r_i) \!= \!P_t l_i(r_i)\kappa/(4\pi)$, the PU EMFE is defined as $\mathcal{S}\!=\!\sum_{i \in \Psi} \Bar{S}_{r}(r_i)\,G_{i}\, |h_i|^2$. 

\subsection{Antenna Pattern Model}
\label{ssec:gain}
The normalized gain of one ULA with $N$ omnidirectional antenna elements and half-wavelength spacing is given~by
{\smalltonormalsize
\begin{equation}
    G_{act}(\varphi) = \frac{\sin^2\left(\frac{\pi\,N}{2}\,\sin(\varphi)\right)}{N^2\,\sin^2\left(\frac{\pi}{2}\,\sin(\varphi)\right)}
\end{equation}} where $\varphi \in [-\pi/3, \pi/3[$ and $\varphi = 0$ corresponds to the maximal gain of the main lobe $G_{act}(0) = 1$. This function yields intractable mathematical expressions when calculating performance metrics. We introduce the multi-cosine antenna pattern defined as
{\smalltonormalsize
\begin{equation}\label{eq:Gmc}G_{mc}(\varphi) = \left\{
    \begin{array}{ll}
    \!\cos^2\!\left(\frac{N \pi \varphi}{4}\right)  \!   & \text{if  } |\varphi|\leq 2/N\\
    \!\chi_k \sin^2\!\left(\frac{N \pi \varphi}{2}\right)  \!   & \text{if  } \frac{2 k}{N}\!\leq\!|\varphi|\!\leq \!\frac{2 k+2}{N}
    \end{array}\right.
\end{equation}} where $\chi_k = \frac{\sin^2(N x_k)}{N^2 \sin^2(x_k)}$ is the extrema of the $k$th side lobe of the theoretical gain function, with $0 \leq k \leq k_{\textrm{\normalfont max}}$ and $\chi_0\! = \!1$. The choice of $k_{\textrm{\normalfont max}}$ is flexible but should remain below $\floor*{N \sqrt{3}/4-1}$ to prevent side lobes from extending beyond each ULA's sector. The values of $x_k$ are well approached by the ordered positive solutions of $N \tan(x) = \tan(N x)$.
\begin{figure}
    \centering
    \includegraphics[width=0.8\linewidth, trim={3cm, 9cm, 3cm, 9.5cm}, clip]{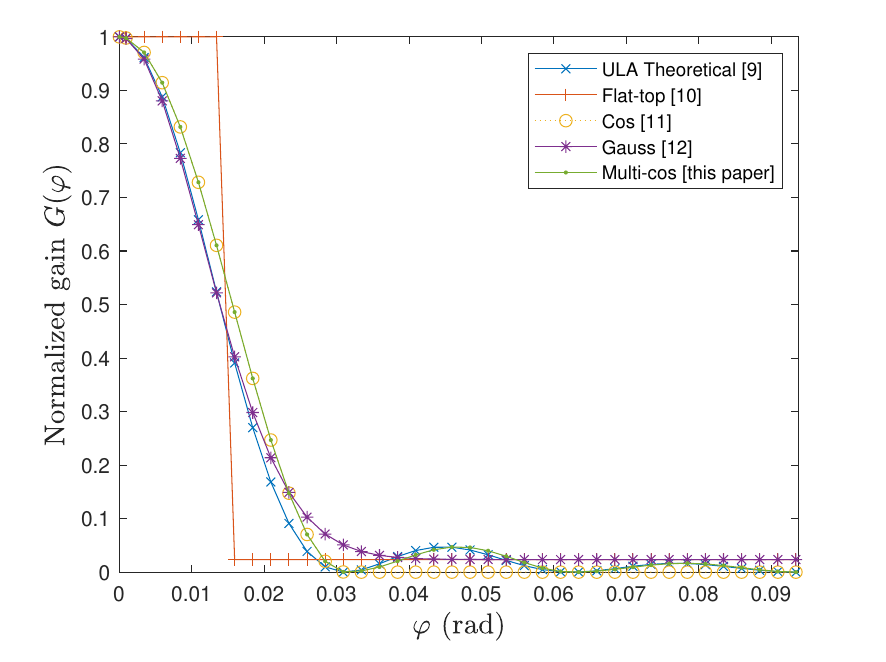}
    \caption{Theoretical ULA antenna gain and its approximations}
    \label{fig:gain_Globecom}
\end{figure}

\section{Mathematical Results}
\label{sec:math}
Subsection~\ref{ssec:cdf_EMFE} introduces the CDF of EMFE of a PU, conditioned on a realization of the PPP, and its associated characteristic function (CF). Subsection~\ref{ssec:meta} defines the meta distribution, and then presents an approximation using a beta distribution obtained through moment-matching, and provides mathematical expressions for the first- and second-order moments of the conditional EMFE CDF.

\subsection{CDF of EMFE Conditioned on $\Psi$}
\label{ssec:cdf_EMFE}

\begin{theorem}\label{th:CDF_cond}
    The CDF of EMFE conditioned on the PPP $\Psi$ for the propagation model in \eqref{eq:model} and the antenna gain in \eqref{eq:Gmc} can be written by
{\smalltonormalsize
\begin{equation}\label{eq:CDF_cond}
    F_{emf}\left(T_{e}|\Psi\right) = \frac{1}{2}-\frac{1}{\pi}\int_{0}^{\infty} \textrm{\normalfont Im}\left[\phi_{E}(q|\Psi)\,e^{-jqT_{e}}\right]\,q^{-1}\,dq
\end{equation}}
where 
{\smalltonormalsize
\begin{equation}\label{eq:Phi_E_def}
    \phi_{E}(q|\Psi) = \prod_{i \in \Psi}\ \! \left(\!1+\frac{6}{N\pi} \!\Bigg(\sum_{k = 0}^{k_{\textrm{\normalfont max}}}\!\frac{\zeta_k(r_i, q)}{\pi}\! -(1+k_{\textrm{\normalfont max}})\Bigg)\right)
\end{equation}} is the conditional CF of the EMFE of the PU and
{\smalltonormalsize
\begin{equation}\label{eq:zeta_final}
    \zeta_k(r_i, q) =  \frac{\pi\sum_{l = 0}^{m-1} \mybinom[0.7]{-1/2}{l}\mybinom[0.7]{m-1}{l}\left(j\frac{q \chi_k\Bar{S}_r(r_i)}{m}\right)^l}{(1-j\frac{q \chi_k\Bar{S}_r(r_i)}{m})^{m-\sfrac{1}{2}}}.
\end{equation}}
\end{theorem}
\begin{proof}
The expression in \eqref{eq:CDF_cond} is derived by applying Gil-Pelaez' theorem. The proof of \eqref{eq:Phi_E_def} is available in Appendix~\ref{app:1}. 
\end{proof}

\subsection{Meta Distribution of EMFE and Beta-Approximation}
\label{ssec:meta}
The meta distribution of the EMFE \eqref{eq:meta_def} can be rewritten as $\mathbb P_\Psi \left[\ln(F_{emf}(T_e|\Psi)) > \ln(s)\right]$. Using the Gil-Pelaez theorem, it follows that 
{\smalltonormalsize
\begin{equation}\label{eq:meta_GP}
    \mathcal{F}_{F_{emf}}(T_{e}, s) = \frac{1}{2}+\frac{1}{\pi} \int_0^{\infty} \textrm{\normalfont Im}\left[e^{-j q \ln(s)}\mathcal{M}_{jq}(T_e)\right] q^{-1}\,dq
\end{equation}}where $\mathcal{M}_{jq}(T_e) = \mathbb E_\Psi\left[\left(F_{emf}(T_e|\Psi)\right)^{jq}\right]$ is the moment of order $jq$ of the conditional CDF of EMFE. Unlike the meta distribution of the SINR, the moments are not mathematically tractable for any value of~$q$. Although an exact derivation of the meta distribution of EMFE is not possible, a close approximation is provided by the beta distribution, leveraging the methodology of \cite{uplink_meta}, based on the knowledge of the first two moments of the conditional CDF, for which a tractable expression can be found.
\begin{proposition}\label{prop:beta_approx}The meta distribution of EMFE for the propagation model in \eqref{eq:model} and the antenna gain in \eqref{eq:Gmc}, \eqref{eq:meta_GP}, can be approximated by
{\smalltonormalsize
\begin{equation}\label{eq:approx_beta}
    \Tilde{\mathcal{F}}_{F_{emf}}(T_{e}, s) = 1-\mathcal{I}_s\left[\Tilde{\alpha}_{\mathcal{S}}(T_e), \Tilde{\beta}_{\mathcal{S}}(T_e)\right]
\end{equation}}
where $\mathcal{I}_s(a,b) = \frac{B(s;a,b)}{B(a,b)}$ is the regularized incomplete beta function, $B(s;a,b) = \int_0^s t^{a-1}(1-t)^{b-1}$ is the incomplete beta function, $B(a,b) = B(1;a,b)$ the beta function and $\Tilde{\alpha}_{\mathcal{S}}(T_e)$ and $\Tilde{\beta}_{\mathcal{S}}(T_e)$ are the parameters of the beta distribution obtained via moment matching. Under the condition $\mathcal{M}_{1}(T_e)< \mathcal{M}_{2}(T_e)$, their expressions are given by
{\smalltonormalsize
\begin{equation}
    \Tilde{\alpha}_{\mathcal{S}}(T_e) = \mathcal{M}_{1}(T_e) \left(\frac{\mathcal{M}_{1}(T_e)\left(1-\mathcal{M}_{1}(T_e)\right)}{\mathcal{M}_{2}(T_e)-\mathcal{M}_{2}^2(T_e)}-1\right)
\end{equation}}and
{\smalltonormalsize
\begin{equation}
    \Tilde{\beta}_{\mathcal{S}}(T_e) = \left(1-\mathcal{M}_{1}(T_e)\right) \left(\frac{\mathcal{M}_{1}(T_e)\left(1-\mathcal{M}_{1}(T_e)\right)}{\mathcal{M}_{2}(T_e)-\mathcal{M}_{2}^2(T_e)}-1\right).
\end{equation}}
\end{proposition}

\begin{theorem}\label{th:M1}
    The first-order moment of the conditional distribution of EMFE \eqref{eq:CDF_cond} is expressed as
    {\smalltonormalsize
    \begin{equation}
        \mathcal{M}_{1}(T_e) = F_{emf}(T_e) = \frac{1}{2}-\frac{1}{\pi}\int_{0}^{\infty} \textrm{\normalfont Im}\!\left[\phi_{E}(q)\,e^{-jqT_{e}}\right]\,q^{-1}\,dq
    \end{equation}} where the unconditioned CF of EMFE, $\phi_{E}(q)$, is given in \eqref{phiE_final} at the top of the next page. Notation $[f(x)]_{x = a}^{x=b} = f(b)-f(a)$ is employed.
    \begin{figure*}[!h]
{\footnotesize
    \begin{equation}\label{phiE_final}
        \phi_{E}(q) = \exp\left(\frac{-6\lambda (1+k_{\textrm{\normalfont max}})}{N}   [r^2]_{r = r_e}^{r = \tau}-\frac{24\lambda}{N}\sum_{k = 0}^{k_{\textrm{\normalfont max}}}\sum_{l = 0}^{m-1}\Bigg[\frac{ \mybinom[0.7]{-1/2}{l}\mybinom[0.7]{m-1}{l}{(r^2+z^2){}_2F_1\left(1, \scalebox{0.9}{1-\textit{l}+}\frac{2}{\alpha}; \frac{1}{2}\scalebox{0.9}{-\textit{l+m}+}\frac{2}{\alpha}; \frac{m}{jq \chi_k\Bar{S}_r(r)}\right)}}{\left(4+\alpha(2m-2l-1)\right)\left(jq \chi_k\Bar{S}_r(r)/m\right)^{1-l} \left(1-j {q \chi_k\Bar{S}_r(r)}/{m}\right)^{m-\sfrac{3}{2}}}\Bigg]_{r = r_e}^{r = \tau}\right)
    \end{equation}
}
{\scriptsize
\begin{multline}\label{eq:gamma}
    \gamma_\pm(q, q') = \exp\left(\frac{12 \lambda}{N}\Bigg[(1+k_{\textrm{\normalfont max}})\left(\frac{3(1+k_{\textrm{\normalfont max}})}{N\pi}-1\right)r^2\right.\\
    \left.-\sum_{k = 0}^{k_{\textrm{\normalfont max}}}\!\sum_{l = 0}^{m-1}\!\mybinom[0.7]{-1/2}{l}\mybinom[0.7]{m-1}{l}\left(\frac{6}{N\pi}\!\sum_{p = 0}^{k_{\textrm{\normalfont max}}} \!\sum_{l' = 0}^{m-1} \!\mybinom[0.7]{-1/2}{l'}\mybinom[0.7]{m-1}{l'} \frac{r^{2+\alpha(2m-1)}  F_1\!\left(2m\!-\!1\!-\!l\!-\!l'\!+\!\frac{2}{\alpha}, \frac{2m-1}{2}, \frac{2m-1}{2}, 2m\!-\!l-l'\!+\!\frac{2}{\alpha}, \frac{-j m}{ q \chi_k \Bar{S}_r(r)}, \frac{\mp j m}{ q' \chi_k \Bar{S}_r(r)}\!\right)}{\left(2+\alpha(2m-1-l-l')\right)( j q \chi_k \Bar{S}_r(r)/m)^{m-l-\sfrac{1}{2}}(\pm  j q' \chi_k \Bar{S}_r(r)/m)^{m-l'-\sfrac{1}{2}}}  \right. \right.\\
    \left.\left.\!+ \!\frac{2{(r^2+z^2)\!\left(1\!-\!\frac{6(1+k_{\textrm{\normalfont max}})}{N\pi}\right)}}{\left(4+\alpha(2m-2l-1)\right)}
    \!\left(\!
    \frac{{}_2F_1\!\left(1, \scalebox{0.9}{1-\textit{l}+}\frac{2}{\alpha}; \frac{1}{2}\scalebox{0.9}{-\textit{l+m}+}\frac{2}{\alpha}; \frac{-j m}{ q \chi_k \Bar{S}_r(r)}\right)}{\left( j q \chi_k \Bar{S}_r(r)/m\right)^{1-l} \left(1- j q \chi_k \Bar{S}_r(r)/m\right)^{m-\sfrac{3}{2}}}\!+\!\frac{{}_2F_1\!\left(1, \scalebox{0.9}{1-\textit{l}+}\frac{2}{\alpha}; \frac{1}{2}\scalebox{0.9}{-\textit{l+m}+}\frac{2}{\alpha}; \frac{\mp j m}{ q' \chi_k \Bar{S}_r(r)}\right)}{ \left(\pm j q' \chi_k \Bar{S}_r(r)/m\right)^{1-l}  \left(1\mp  j q' \chi_k \Bar{S}_r(r)/m\right)^{m-\sfrac{3}{2}}}\right)
    \right)\!\Bigg]_{r = r_e}^{r = \tau}\!\right)\!
    \end{multline}
}

\hrulefill
\end{figure*}
\end{theorem}
\begin{proof}
The proof is provided in Appendix~\ref{app:2}.
\end{proof}

\begin{theorem}\label{th:M2} The second-order moment of the conditional distribution of EMFE \eqref{eq:CDF_cond} is given by
{\smalltonormalsize
\begin{equation}\label{eq:M2}
    \mathcal{M}_{2}(T_{e}) = -{1}/{4} + F_{emf}(T_{e})
    + {\pi^{-2}}\, {\Omega(T_e)}
\end{equation}}
where 
{\smalltonormalsize
\begin{equation}\label{eq:Omega}
    \Omega(T_e) = \int_{0}^{\infty}\int_{0}^{\infty}{\omega(T_e, q, q')}\,q^{-1}\,q'^{-1}\,dq\,dq',
\end{equation}}
{\small
\begin{equation}
        \omega(q, q';T_{e}) = \frac{1}{2}\,\textrm{\normalfont Re}\!\left[\gamma_-(q, q')\,e^{-jT_{e}(q-q')}\!+\!\gamma_+(q, q')\,e^{-jT_{e}(q+q')}\right]
\end{equation}}and $\gamma_\pm(q, q')$ is given in \eqref{eq:gamma} at the top of the next page.
\end{theorem}
\begin{proof}
The proof is provided in Appendix~\ref{app:3}.
\end{proof}
\section{Numerical Results}
\label{sec:numerical_results}
This section analyzes the meta distribution of the EMFE of the PU based on the system model described in Section~\ref{sec:system_model}. Unless specified otherwise, the parameters are set as follows: $f = \numprint[GHz]{3.5}$, $\lambda = \numprint[BS/km^2]{10}$, $P_t = \numprint[dBm/m^2]{66}$, $N$ = 64, $\alpha$ = 3.25, $z$ = \numprint[m]{30}, $\tau$ = \numprint[km]{3}, $m = 3$ and $k_{\textrm{\normalfont max}} = 9$.

The first- and second-order moments, as given by Theorems~\ref{th:M1} and \ref{th:M2}, are shown in Fig.~\ref{fig:BF_meta_moments} and compared to Monte-Carlo (MC) simulations with $10^4 \times 10^4$ realizations. The close agreement confirms the validity of the mathematical approach. The first-order moment corresponds to the CDF of EMFE, while the second-order moment characterizes deviation from this CDF. Maximum deviation occurs around the median EMFE (-38 dBm/m$^2$), while deviations are minimal at both ends of the CDF of EMFE. This indicates low uncertainty regarding extreme probability values.
\begin{figure}
    \centering
    \includegraphics[width=0.8\linewidth, trim={3cm, 9.2cm, 3cm, 10cm}, clip]{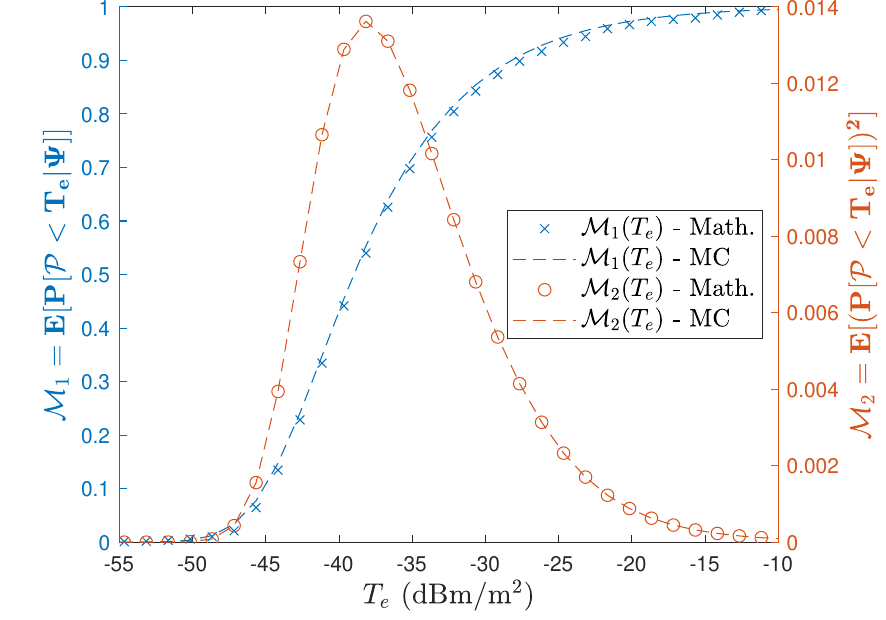}
    \caption{First- and second-order moments of the CDF of EMFE of PUs computed from the mathematical expressions (Math.) or from MC simulations (MC). The first-order moment corresponds to the CDF of EMFE of PUs.}
    \label{fig:BF_meta_moments}
\end{figure}

In Fig.\ref{fig:BF_meta_act_Te}, the meta distribution for various EMFE limits is depicted numerically ("Exact"), alongside its beta approximation ("Approx."), derived from Proposition~\ref{prop:beta_approx}. The significant alignment between the curves highlights the precision of this approximation. Despite the complexity of the mathematical expression of the moments, using SG proves advantageous due to the reduced computational time compared to MC simulations.
At $T_e \!= \!\numprint[dBm/m^2]{-32}$, over 90\% of PU locations have a 69\% probability of experiencing EMFE below this threshold. In other words, 90\% of PUs encounter EMFE below $T_e \!= \!\numprint[dBm/m^2]{-32}$ at least 69\% of the time. Increasing the threshold to $T_e = \numprint[dBm/m^2]{-23}$ sees the same proportion of PUs below it 90\% of the time, contrasting with only 22\% if reduced to $T_e \!= \!\numprint[dBm/m^2]{-41}$. Similarly, for 90\% of users being below the threshold 90\% of the time, $T_e\! \geq \!\numprint[dBm/m^2]{-32}$ is required (equivalent to 0.05~V/m  \cite{GontierAccess}) as shown by the green rectangle.
\begin{figure}
    \centering
    \includegraphics[width=0.8\linewidth, trim={3cm, 9.5cm, 3cm, 9.5cm}, clip]{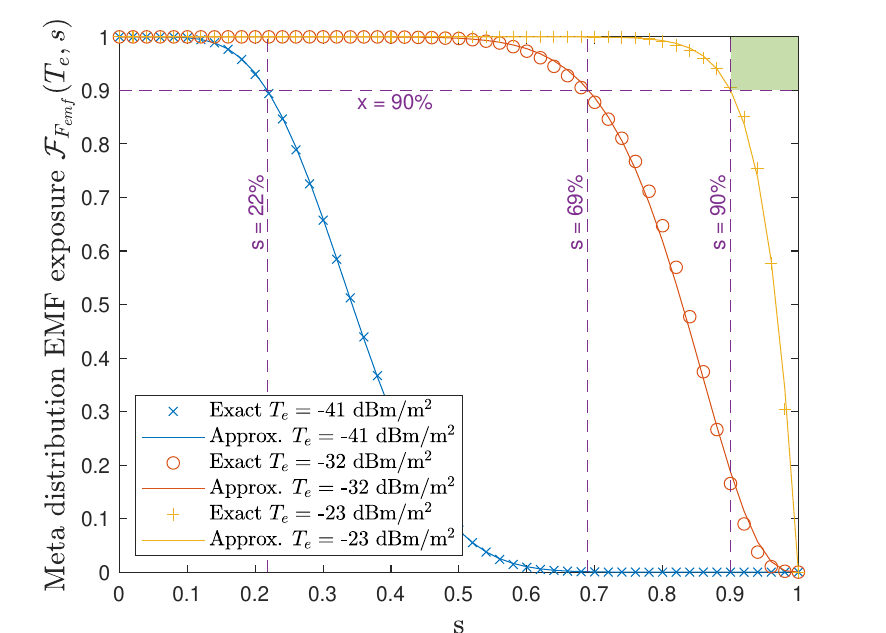}
    \caption{Meta distribution of EMFE of PUs for different EMFE limits $T_e$ and approximation by a beta distribution ($N = 64$)}
    \label{fig:BF_meta_act_Te}
\end{figure}

Finally, Fig.~\ref{fig:BF_meta_act_N} illustrates the meta distribution for various number of antenna elements in the ULA at $T_e = \numprint[dBm/m^2]{-32}$. The influence of $N$ on PU EMFE is notable, with 90\% below the threshold just 14\% of the time for $N = 16$, contrasting with 69\% for $N = 64$, and 92\% for $N = 256$.

\begin{figure}
    \centering
    \includegraphics[width=0.8\linewidth, trim={3cm, 9.5cm, 3cm, 10cm}, clip]{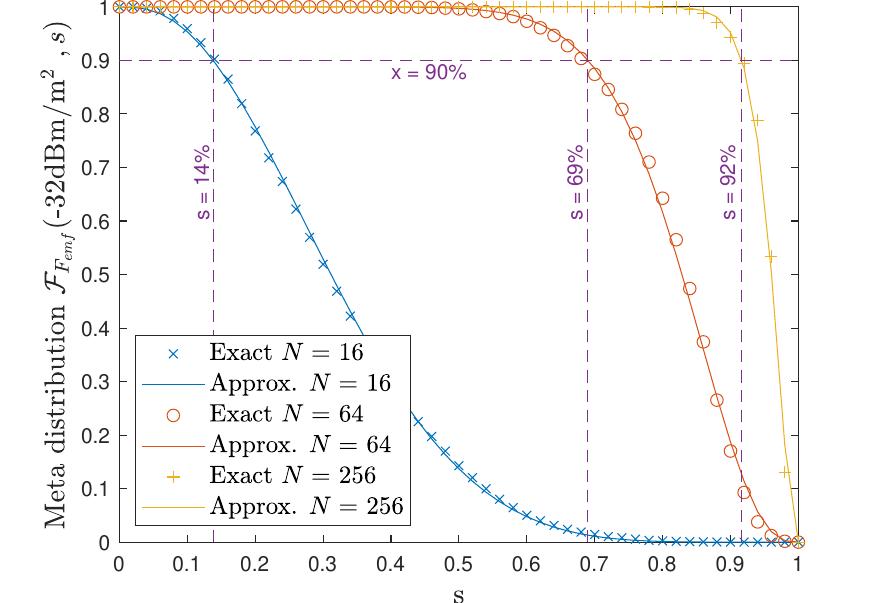}
    \caption{Meta distribution of EMFE of PUs for different values of $N$ and approximation by a beta distribution ($T_e = \numprint[dBm/m^2]{-32}$)}
    \label{fig:BF_meta_act_N}
\end{figure}

\section{Conclusion}
\label{sec:conclusion}
In conclusion, this paper has presented a comprehensive analysis of the meta distribution of the EMFE experienced by the PU within a PPP network with Nakagami-$m$ fading and DBF, implemented using a multi-cosine gain model. The meta distribution has been approximated by a beta distribution via the moment-matching method, for which the first- and second-order moments of the meta distribution have been calculated. The proposed approximation has been validated against numerical results. The analysis sheds light on the impact of various system parameters, such as the EMFE limit and the number of antennas in the ULA, on the meta distribution.
{\appendices
\section{Proof of Theorem~\ref{th:CDF_cond}}
\label{app:1}The CF of the EMFE conditioned on $\Psi$ is defined as
{\small
\begin{align}
\begin{split}
    &\phi_{E}(q|\Psi) = \mathbb E_{\xi, |h|}\Big[\exp\Big(j q \sum_{i \in \Psi} P_r(r_i)\Big)\Big|\Psi\Big] \\
    &\stackrel{(a)}{=} \prod_{i \in \Psi}\mathbb E_{\xi_i, |h_i|}\left[\exp\left(j q  \Bar{S}_r(r_i)G(\xi_i)|h_i|^2\right)\right] := \prod_{i \in \Psi}\phi_{E}(q|X_i).
\end{split}
\end{align}}Let $\phi_{E}(q|X_i)$ be the CF of the EMFE caused by the $i$th BS. Due to the two-by-two independence of the $\xi_i$ and $h_i$ for all $i$, the expectation operators and the product can be interchanged $(a)$. Considering the Nakagami-$m$ distribution of $|h_i|$, applying $\mathbb E_{|h_i|}$ yields:
{\smalltonormalsize
\begin{equation}\label{eq:phiE}
    \phi_{E}(q|X_i) = \mathbb E_{\xi_i}\left[\left(1-j q \Bar{S}_r(r_i)G(\xi_i)/m\right)^{-m}\right].
\end{equation}}
Expanding \eqref{eq:phiE} using \eqref{eq:Gmc} results in:
{\footnotesize
    \begin{multline}\label{eq:phiE2}
    \phi_{E}(q|X_i) = \frac{3}{\pi} \left(\int_0^{\frac{2}{N}} \left(1-j \frac{q \Bar{S}_r(r_i)}{m}\cos^2\left(\frac{N\pi \xi}{4}\right)\right)^{-m}d\xi\right.\\
    \!+\left.\sum_{k = 1}^{k_{\textrm{\normalfont max}}}\!\int_{\frac{2 i}{N}}^{\frac{2i+2}{N}}\! \frac{1}{\left(1-j \chi_k\frac{q \Bar{S}_r(r_i)}{m}\sin^2\left(\frac{N\pi \xi}{2}\right)\right)^{m}}d\xi+\!\int_{\frac{2i+2}{N}}^{\frac{\pi}{3}}\!d\xi\right)\!.\!
\end{multline}}
The rest of the proof, focusing on the integration of \eqref{eq:phiE2}, is available in Appendix~\ref{app:1}.
Starting from \eqref{eq:phiE2} and using the change of variable $y \to N\pi \xi/2$ for the first integral and $y \to N\pi \xi$ for the following integrals, using the relationship $\sin(x) = \cos(\pi/2-x)$ and exploiting the symmetry of $\cos^2(x)$, \eqref{eq:phiE2} can be reformulated~as
{\small
\begin{equation}
    \phi_{E}(q|X_i) = \frac{6}{N\pi^2} \sum_{k = 0}^{k_{\textrm{\normalfont max}}}\zeta_k(r_i, q)+1-\frac{6}{\pi N} (k_{\textrm{\normalfont max}}+1)
\end{equation}} with $\zeta_k(r_i, q) = \int_{0}^{\pi} \!\left(1-j \chi_k\frac{q \Bar{S}_r(r_i)}{m}\cos^2\left(y/2\right)\right)^{-m}dy$.
To solve the integral $\zeta_k(r_i, q)$, the denominator and numerator first have to be multiplied by $\cos^{-2m}\left(y/2\right)$. Second a change of variable $u \to \tan(y/2)$ has to be used. Third, a new change of variable $\tan(x) = u /\sqrt{1-j \chi_k{q \Bar{S}_r(r_i)}/{m}}$. Forth, the relationship $\sec^2(x) = 1+\tan^2(x)$ has to be used. Letting $a = {q \chi_k\Bar{S}_r(r_i)}/{m}$, this gives
{\smalltonormalsize
\begin{align}
\begin{split}
    \zeta_k(r_i, q) = \frac{1}{({1-ja})^{-1}} \int_0^{\pi} \left(\frac{\tan^2(x)+\frac{1}{1-j a}}{\tan^2(x)+1}\right)^{m-1} dx \\
    = \frac{1}{\sqrt{1-ja}} \int_0^{\pi} \left(1+\frac{ja}{1-ja}\frac{1}{\tan^2(x)+1}\right)^{m-1} dx.
\end{split}
\end{align}}
Using the binomial theorem, this can be rewritten as 
{\small
\begin{equation}\label{eq:zeta}
    \zeta_k(r_i, q) =  \sum_{l = 0}^{m-1} \mybinom[0.7]{m-1}{l}\int_0^{\pi}\!\frac{(ja)^l}{(1-ja)^{l+\sfrac{1}{2}}\left(\tan^2(x)+1\right)^{l}}  dx.
\end{equation}}
At last, we use the result
$
    \int_0^\pi \!\left(\!-\!\tan^2(x)\!-\!1\right)^{-l} \!dx \!= \!\mybinom[0.6]{-1/2}{l}\pi
$ to solve \eqref{eq:zeta}.
After some simplifications, \eqref{eq:zeta_final} is obtained.

\section{Proof of Theorem~\ref{th:M1}}
\label{app:2}
    Since $\mathcal{M}_{1}(T_e) = \mathbb E_{\Psi}\left[F_{emf}(T_e|\Psi)\right]$ with $F_{emf}(T_e|\Psi)$ given in \eqref{eq:CDF_cond}, it is enough to observe that the expectation operator can be interchanged with the integral and the imaginary part operator, its only effect being on $\phi_{E}(q|\Psi)$. Then, $\phi_{E}(q) = \mathbb E\left[\phi_{E}(q|\Psi)\right]$ with $\phi_{E}(q|\Psi)$ given in \eqref{eq:Phi_E_def}. The unconditioned CF is then obtained by applying the probability generating functional \cite{Baccelli1997StochasticGA}:
    {\footnotesizetosmall
    \begin{align}\label{eq:phiE_proof}
    \begin{split}
        &\phi_{E}(q) = \exp\left(\frac{-12\lambda}{N}  \int_{r_e}^{\tau}\left( 1+k_{\textrm{\normalfont max}}-\sum_{k = 0}^{k_{\textrm{\normalfont max}}}\frac{\zeta_k(r, q)}{\pi}\right)r \,dr\right)\\
        &= \exp\!\left(\frac{-6\lambda (1+k_{\textrm{\normalfont max}})}{N}   [r^2]_{r = r_e}^{r = \tau}\!+\!\frac{12\lambda}{N}\sum_{k = 0}^{k_{\textrm{\normalfont max}}}\!\int_{r_e}^{\tau}\!\frac{\zeta_k(r, q)}{\pi}r \,dr\right).
    \end{split}
    \end{align}}To integrate $\zeta_k(r, q)$, we have to focus on the integral
     {\small
    \begin{equation}
        I = \int_{r_e}^{\tau}\!\left(j a \left(r^2+z^2\right)^{-\sfrac{\alpha}{2}}\right)^l \!\left(1-j a \left(r^2+z^2\right)^{-\sfrac{\alpha}{2}}\right)^{-m+\sfrac{1}{2}}r\,dr.
    \end{equation}}
    By writing the integral from $0$ to $x$ $I(x)$ such that $I = [I(x)]_{x = r_e}^{x = \tau}$ and by using the change of variable $u \to (r^2+z^2)^{-\sfrac{\alpha}{2}}(x^2+z^2)^{\sfrac{\alpha}{2}}$, it is obtained that 
    {\smalltonormalsize
    \begin{equation}
        I = \Big[\dfrac{-\left(j a\right)^l }{\alpha}  \int_{0}^{1}\frac{(x^2+z^2)^{1-\sfrac{l\alpha}{2}}u^{l-1-\sfrac{2}{\alpha}}}{\left(1-j \,a\, u (x^2+z^2)^{-\sfrac{\alpha}{2}}\right)^{m-\sfrac{1}{2}}}du\Big]_{x = r_e}^{x = \tau}.
    \end{equation}}
    Using the definition of the hypergeometric function ${}_2F_1(\Tilde{a}, \Tilde{b}; \Tilde{c}; \Tilde{d})$ for $\real{\Tilde{c}}>\real{\Tilde{b}}>0$ and $|\arg(1-\Tilde{d})| < \pi$,
    {\smalltonormalsize
    \begin{equation}
        {}_2F_1(\Tilde{a}, \Tilde{b}; \Tilde{c}; \Tilde{d}) = \frac{\Gamma(\Tilde{c})}{\Gamma(\Tilde{b})\Gamma(\Tilde{c}-\Tilde{b})}\int_0^1 \frac{t^{\Tilde{b}-1}(1-t)^{\Tilde{c}-\Tilde{b}-1}}{(1-t \Tilde{d})^{\Tilde{a}}} dt,
    \end{equation}}
    and after some simplifications, we have
    {\footnotesize
    \begin{equation}\label{eq:zeta_integrated}
        I\! = \!\Bigg[\frac{(x^2+z^2)^{1-(l-1)\frac{\alpha}{2}}{}_2F_1\!\left(1, \scalebox{0.9}{1-\textit{l}+}\frac{2}{\alpha}; \frac{1}{2}\scalebox{0.9}{-\textit{l+m}+}\frac{2}{\alpha}; \frac{-j}{a(x^2+z^2)^{-\frac{\alpha}{2}}}\right)}{ja \left(1-ja(x^2+z^2)^{-\sfrac{\alpha}{2}}\right)^{m-\sfrac{3}{2}}\left(2+\sfrac{\alpha}{2}(2m-2l-1)\right)}\Bigg]_{x = r_e}^{x = \tau}
    \end{equation}}
Using \eqref{eq:zeta_final} and inserting \eqref{eq:zeta_integrated} in \eqref{eq:phiE_proof} gives \eqref{phiE_final}.

\section{Proof of Theorem~\ref{th:M2}}
\label{app:3}
     The first step consists of expanding $\mathcal{M}_{2}(T_e)$ by using Theorem~\ref{th:CDF_cond} and distributing the expectation operator. This gives
 {\small
\begin{multline}\label{eq:M2_first}
    \mathcal{M}_{2}(T_{e}) = \frac{1}{4} - \mathbb E_{\Psi}\left[\frac{1}{\pi}\int_{0}^{\infty} \textrm{\normalfont Im}\left[\phi_{E}(q|\Psi)\,e^{-jqT_{e}}\right]\,q^{-1}\,dq\right]\\
    + \mathbb E_{\Psi}\left[\frac{1}{\pi^2}\left(\int_{0}^{\infty} \textrm{\normalfont Im}\left[\phi_{E}(q|\Psi)\,e^{-jqT_{e}}\right]\,q^{-1}\,dq\right)\right.\\
    \left.\times\left(\int_{0}^{\infty} \textrm{\normalfont Im}\left[\phi_{E}(q'|\Psi)\,e^{-jq'T_{e}}\right]\,q'^{-1}\,dq'\right)\right].
\end{multline}}The second term in \eqref{eq:M2_first} corresponds to $\mathcal{M}_{1}(T_{e})-1/2$. By letting $\Omega(T_e)$ be the third term in \eqref{eq:M2_first}, \eqref{eq:M2} is obtained. 

The second step consists of interchanging the integrals and the expectation operator in the expression of $\Omega(T_e)$, so that it can be written as in \eqref{eq:Omega} with $
    \omega(T_e, q, q') \!= \!\mathbb E_{\Psi}[\textrm{\normalfont Im}[\phi_{E}(q|\Psi)e^{-jqT_{e}}]\textrm{\normalfont Im}[\phi_{E}(q'|\Psi) e^{-jq'T_{e}}]]$.

The third step focuses on the expansion of $\omega$. By using $\textrm{\normalfont Im}[x]\! = \!(x-\Bar{x})/{2}$ and $\textrm{\normalfont Re}[x]\! = \!({x+\Bar{x}})/{2}$ and by writing $v_1 \!=\! \exp(-jqT_{e})$ and $v_2\! = \!\exp(-jq'T_{e})$, it is obtained that $
    \omega(T_e, q, q') 
    = \frac{1}{2}\,\textrm{\normalfont Re}\left[\gamma_+(q, q')\,v_1\,v_2\right]-\frac{1}{2}\,\textrm{\normalfont Re}\left[\gamma_-(q, q')\,v_1\,\Bar{v_2}\right]
$
with $\gamma_+(q, q')\! = \!\mathbb E_{\Psi}\left[\phi_{E}\left(q | \Psi\right)\,\phi_{E}\left(q' | \Psi\right)\right]$ and $\gamma_-(q, q') = \mathbb E_{\Psi}\left[\phi_{E}\left(q | \Psi\right)\,\Bar{\phi_{E}}\left(q' | \Psi\right)\right]$. 
The last step is the resolution of $\gamma_+(q, q')$ and $\gamma_{-}(q, q')$, 
By replacing the CFs by their expression given in \eqref{eq:Phi_E_def}, by distributing the product and by using the probability generating functional, we obtain for $\gamma_+(q, q')$
{\footnotesize
\begin{align}\label{eq:gamma+_proof}
\begin{split}
    &\gamma_+(q, q') = \mathbb E_{\Psi}\left[\prod_{X_i \in \Psi}\phi_{E}(q|X_i)\phi_{E}(q'|X_i)\right]\\
    &= \exp\left(\int_{r_e}^{\tau}\frac{24(1+k_{\textrm{\normalfont max}})\lambda}{N}\left(\frac{3(1+k_{\textrm{\normalfont max}})}{N\pi}-1\right)r\,dr\right.\\
    &+\left.\int_{r_e}^{\tau}\frac{12\lambda}{N}\left(1-\frac{6(1+k_{\textrm{\normalfont max}})}{N\pi}\right) \sum_{k = 0}^{k_{\textrm{\normalfont max}}}\frac{\zeta_k(r, q)+\zeta_k(r, q')}{\pi}r\,dr\right.\\
    &\left.+\int_{r_e}^{\tau}\frac{72\lambda}{N^2\pi}\sum_{k = 0}^{k_{\textrm{\normalfont max}}}\sum_{p = 0}^{k_{\textrm{\normalfont max}}}\frac{\zeta_k(r, q)\zeta_p(r, q')}{\pi^2}r\,dr\right).
\end{split}
\end{align}}
We first interchange the integral and sums in \eqref{eq:gamma+_proof}. The integration of $\zeta_k(r, q)$ has been completed in Appendix~\ref{app:2}. The integration of the product $\zeta_k(r, q)\zeta_p(r, q')$ is only feasible for $z = 0$. Employing the integration method from Appendix~\ref{app:2}, with $a = q \chi_k \Bar{S}_r(r_i)/m$ and $b = q' \chi_k \Bar{S}_r(r_i)/m$, and using the first Appell function, a generalization of the hypergeometric function ${}_2F_1$, defined as
{\footnotesizetosmall
\begin{equation}
    F_1(\Tilde{a}; \Tilde{b}_1, \Tilde{b}_2; \Tilde{c}; \Tilde{z}_1, \Tilde{z}_2) = \frac{\Gamma(\Tilde{c})}{\Gamma(\Tilde{a})\Gamma(\Tilde{c}-\Tilde{a})}\int_0^1 \frac{t^{\Tilde{a}-1}(1-t)^{\Tilde{c}-\Tilde{a}-1}}{(1-t \Tilde{z}_1)^{\Tilde{b}_1}(1-t \Tilde{z}_2)^{\Tilde{b}_2}} dt,
\end{equation}}
we obtain
{\scriptsize
\begin{align}
\begin{split}\label{eq:integration_zeta_zeta}
    &\int_{r_e}^{\tau} \frac{r^{-(l+l') \alpha} r}{\left((1- j a r^{-\alpha})(1-j b r^{-\alpha})\right)^{m-\sfrac{1}{2}}} dr \\
    &\quad= -\left[\frac{r^{2+\alpha(2m-1-l-l')} (-1)^{m+1}}{\left(2+\alpha(2m-1-l-l')\right)(j a)^{m-l-\sfrac{1}{2}}(j b)^{m-l-\sfrac{1}{2}}}\right. \times \\
    &\!\left. F_1\!\left(2m\!-\!1\!-\!l\!-\!l'\!+\!\frac{2}{\alpha}; m\!-\!\frac{1}{2}, m\!-\!\frac{1}{2}; 2m\!-\!\!l\!-\!l'\!+\!\frac{2}{\alpha}; \frac{r^\alpha}{j a}, \frac{r^\alpha}{j b}\!\right)\right]_{r = r_e}^{r = \tau}\!.
\end{split}
\end{align}}
Finally, by inserting \eqref{eq:integration_zeta_zeta} in \eqref{eq:gamma+_proof}, and applying the same reasoning to $\gamma_{-}(q,q')$, we obtain \eqref{eq:gamma}.
}
\bibliographystyle{IEEEtran}
\bibliography{bibli}
\vfill
\end{document}